\documentclass[11pt]{article}
\usepackage[utf8]{inputenc}
\usepackage{style}
\usepackage{framed}
\usepackage[parfill]{parskip}
\usepackage{setspace}
\usepackage{hyperref}

\allowdisplaybreaks

\title{Sherali-Adams Integrality Gaps \\ Matching the Log-Density Threshold}

\author{
Eden Chlamt\'{a}\v{c}\thanks{Email: \texttt{chlamtac@cs.bgu.ac.il}.} \vspace{-0.5em}\\
Ben-Gurion University
\and
Pasin Manurangsi\thanks{Email: \texttt{pasin@berkeley.edu}. Part of this work was done while the author was visiting Ben-Gurion University.} \vspace{-0.5em}\\
UC Berkeley
}

\begin{document}

\begin{titlepage}
\def\thepage{}
\maketitle

\begin{abstract}
The log-density method is a powerful algorithmic framework which in recent years has given rise to the best-known approximations for a variety of problems, including Densest-$k$-Subgraph and Bipartite Small Set Vertex Expansion. These approximations have been conjectured to be optimal based on various instantiations of a general conjecture: that it is hard to distinguish a fully random combinatorial structure from one which contains a similar planted sub-structure with the same ``log-density''.

We bolster this conjecture by showing that in a random hypergraph with edge probability $n^{-\alpha}$, $\tilde\Omega(\log n)$ rounds of Sherali-Adams with cannot rule out the existence of a $k$-subhypergraph with edge density $k^{-\alpha-o(1)}$, for any $k$ and $\alpha$. This holds even when the bound on the objective function is lifted. This gives strong integrality gaps which exactly match the gap in the above distinguishing problems, as well as the best-known approximations, for Densest $k$-Subgraph, Smallest $p$-Edge Subgraph, their hypergraph extensions, and Small Set Bipartite Vertex Expansion (or equivalently, Minimum $p$-Union). Previously, such integrality gaps were known only for Densest $k$-Subgraph for one specific parameter setting.
\end{abstract}

\end{titlepage}\pagenumbering {arabic}



\section{Introduction}

The log-density framework~\cite{BCCFV10} is an emerging technique in approximation algorithms that proposes to understand the problems of interest via an average case study. More specifically, the first step in this framework is to consider a distinguishing problem between a ``random instance'' and a ``random instance with planted solution''. Instead of considering any algorithm, the focus here is on a simple ``witness counting'' algorithm, in which one only counts the number of occurrences of a certain substructure (i.e. witness) of the two instances. One then uses the insights from this simple witness to devise an algorithm for worst case instances. Although at first glance this last step may seem like a large leap (i.e. from simple algorithms for average case instances to more complicated algorithms for worst case instances) and even somewhat implausible, the framework has turned out to be quite effective in tackling a number of problems, including Densest $k$-Subgraph~\cite{BCCFV10}, Lowest Degree 2-Spanner, Smallest $p$-Edge Subgraph (S$p$ES)~\cite{CDK12}, Small Set Bipartite Vertex Expansion (SSBVE)~\cite{CDM17}, Label Cover and 2-CSPs~\cite{CMMV17}.

To be more concrete, let us consider the Densest $k$-Subgraph (D$k$S) problem. In this problem we are given an undirected graph $G = (V, E)$ and an integer $k$, and the goal is to find subgraph of $G$ of at most $k$ vertices that induces a maximum number of edges. The work of~\cite{BCCFV10} considers the following distinguishing problem:
\begin{definition}[Random vs Planted Subgraph Problem] For certain parameters $\alpha,\beta\in(0,1),1<k<n$, given a graph $G$, decide whether it is sampled from the Erd\H os-R\'enyi distribution $\cG(n,n^{-\alpha})$ or from the same distribution, but with a planted subgraph on $k$ randomly chosen vertices, sampled from $\cG(k,k^{-\beta})$.\end{definition}
They noted that simple witness counting algorithms can solve the distinguishing problem w.h.p.\ in time $n^{O(1/\epsilon)}$ whenever $\beta\leqs\alpha-\epsilon$, and fail to distinguish when $\beta\geqs\alpha$. In particular, when $k=n^{\alpha}$, the witness counting algorithm fails to distinguish between two instances, one in which the densest $k$-subgraph has at most $\tilde O(k)$ edges, and one in which it has $\Omega(k^{1+1-\alpha})=\Omega(k\cdot n^{\alpha(1-\alpha)})$ edges. They then describe an algorithm for general instances which achieves this exact tradeoff: when $k=n^{\alpha}$, their algorithm gives an $O(n^{\alpha(1-\alpha)+\epsilon})$-approximation in time $n^{O(1/\epsilon)}$. 

It has been conjectured (see for example~\cite{CDK12, CDM17}) that not only can the above approximation guarantee not be improved for worst-case instances, but even the Random vs Planted distinguishing problem cannot be solved below the same threshold:


\begin{conjecture}[Planted Dense Subgraph Conjecture] \label{conj:planted-subgraph}
For all $0<\alpha <  1$ and sufficiently small $\epsilon>0$, and for all $k\leqs\sqrt{n}$,\footnote{Note that, while for $k \geqs \sqrt{n}$, a spectral algorithm can distinguish the two cases beyond the log-density threshold, the best \emph{worst-case} approximation guarantees are still at the log-density threshold even for this regime.} no polynomial time algorithm can solve the Random vs Planted Subgraph problem for parameters $\beta\geqs\alpha+\epsilon$ with non-negligible probability.
\end{conjecture}

The above conjecture demonstrates another intriguing aspect of the log-density framework: not only that the framework leads to improved algorithms, it also leads to conjectured tight hardness results. For Densest $k$-Subgraph, this should be contrasted with the fact that no NP-hardness of approximation is known even for a factor of say 1.01! Moreover, while inapproximability results for D$k$S are known under stronger complexity assumptions~\cite{Feige02,Kho06,RS10,AAMMW11,BKRW17,M17}, none of them achieves a ratio of the form $n^\delta$ for some $\delta > 0$ and hence they are still very far from giving a tight lower bound as predicted by Conjecture~\ref{conj:planted-subgraph}.

The importance of such tight lower bounds is also amplified by the known connections between Densest $k$-Subgraph and numerous other problems. In particular, there are reductions from D$k$S to many problems, for which either hardness of approximation is not known at all otherwise or the known hardness is very weak compared to known approximations (see e.g.~\cite{HJ06,HJLMRV06,CHK11,ABBG11,ACLR15,CMVZ16,CNW16,Lee17,CZ17,FH17,KKK18}). A hardness for D$k$S with a concrete inapproximability ratio, such one as established by Conjecture~\ref{conj:planted-subgraph}, would imply strong inapproximability results for these problems as well. In addition to applications in hardness of approximation, a variant of the conjecture has also been used in public-key cryptography~\cite{ABW10}.

Despite the aforementioned applications and importance of the conjecture, little progress has been made towards actually justifying it. This is somewhat unsurprising, since, barring few exceptions, it is typically hard to argue for validity of average case assumptions. A common approach to support these hypotheses is by proving lower bounds against restricted classes of algorithms, such as certain LP and SDP relaxations; this has been done, for examples, for the planted clique hypothesis~\cite{DM15,MPW15,HKPRS16,BHKKMP16} and the random 3-SAT hypothesis~\cite{Gri01,Sch08}.

Unfortunately, even on this front, not much is known for Densest $k$-Subgraph. In particular, the only (non-trivial) matching lower bound shown so far is that of Bhaskara et al.~\cite{BCVGZ12}, who showed that $\Omega(\log n/\log\log n)$ rounds of the Sherali-Adams hierarchy~\cite{SA90} have a matching integrality gap only for the case of $\alpha = 1/2$ and $k=\sqrt{n}$, which happens to be the parameter setting maximizing the log-density gap for D$k$S. No tight integrality gaps were known for other parameter settings. The situation is not better for other problems that have been studied in the log-density framework; either no lower bound of this type is known for them at all, or a lower bound is known only for one specific instantiation of parameters. Note that evidence for more specific hardness of this form is crucial in hardness-of-approximation results based on D$k$S and related problems, as an optimal \emph{reduction} may use a different parameter setting than the worst-case setting for D$k$S (e.g.~\cite{CNW16}), or a different problem such as S$p$ES (e.g.~\cite{Lee17}) or SSBVE (e.g.~\cite{CMVZ16}).



\subsection{Our Results}
We show that for every possible parameter setting, Sherali Adams requires a super-constant number of rounds to distinguish between a random hypergraph, and one which contains a subhypergraph with the same log-density:
\begin{theorem}\label{thm:main-informal}[informal; see Theorem~\ref{thm:main-general}] For every $c=c(n)=O(\log\log n)$, every $k\in[n]$, every constant $\epsilon>0$ and every $\epsilon<\alpha<c-1$, for the random Erd\H os-R\'enyi $c$-uniform hypergraph $\cG_c(n,n^{-\alpha})$, there is a $\tilde\Omega(\log n)$-round Sherali-Adams solution for the existence of a subhypergraph of size at most $k$, and at least $k^{c-\alpha-o(1)}$ edges, where both the bound on the size and the bound on the number of edges are lifted.
\end{theorem}

As immediate corollaries of our main theorem, we obtain integrality gaps matching the log-density threshold for a variety of problems, which we outline here. Firstly, for D$k$S, we have an integrality gap matching the approximation guarantee of~\cite{BCCFV10} for every log-density:
\begin{corollary} For any constant $\alpha\in(0,1)$, for $G \sim \cG(n,n^{-\alpha})$ and $k=n^{\alpha}$, $\tilde\Omega(\log n)$ rounds of Sherali-Adams applied to a Densest $k$-Subgraph relaxation with a lifted objective function have an integrality gap of $n^{\alpha(1-\alpha)-o(1)}$ with high probability.
\end{corollary}

In the Smallest $p$-Edge Subgraph problem (S$p$ES), we are given a graph $G$ and a parameter $p$, and are asked to find a subgraph of $G$ with $p$ edges and a minimum number of vertices. For this problem we get integrality gaps matching the approximation guarantee of~\cite{CDK12}:
\begin{corollary} For any constant $\alpha\in(0,1)$, for $G \sim \cG(n,n^{-\alpha})$ and $p=n^{\alpha}$, $\tilde\Omega(\log n)$ rounds of Sherali-Adams applied to a S$p$ES relaxation with a lifted objective function have an integrality gap of $n^{\frac{\alpha(1-\alpha)}{2-\alpha}-o(1)}$ w.h.p. In particular, for $\alpha=2 - \sqrt{2}$, we get an integrality gap of $n^{3-2\sqrt{2}-o(1)}$.
\end{corollary}

We remark that the algorithms in~\cite{BCCFV10,CDK12} can be stated in terms of the Sherali-Adams relaxations used in our work. Hence, we give an essentially tight (up to a sub-polynomial factor) integrality gaps for these relaxations of D$k$S and S$p$ES for every setting of parameters.

Since our main theorem works not only for random graphs but also for random hypergraphs, our integrality gaps extend to the hypergraph variants of D$k$S and S$p$ES. The hypergraph extension of D$k$S, called Densest $k$-Subhypergraph (D$k$SH), is to find, given a $c$-uniform hypergraph $G$ and an integer $k$, a subhypergraph of $k$ vertices that contains as many hyperedges as possible. Our integrality gap for Densest $k$-Subhypergraph is stated below.

\begin{corollary} \label{col:gap-dksh}
For any constants $c \geqs 2$ and $\alpha \in (0, c - 1)$, for $G\sim\cG_c(n,n^{-\alpha})$ and $k=n^{\alpha/(c - 1)}$, $\tilde\Omega(\log n)$ rounds of Sherali-Adams applied to a D$k$SH relaxation with a lifted objective function have an integrality gap of $n^{\alpha(1-\alpha/(c - 1))-o(1)}$ w.h.p. In particular, for $\alpha=(c - 1)/2$, the integrality gap is $n^{(c - 1)/4 - o(1)}$.
\end{corollary}

In the Minimum $p$-Union (M$p$U) problem, the goal is to find $p$ hyperedges whose union is as small as possible in a given $c$-uniform hypergraph. Our integrality gap for M$p$U is stated below.
\begin{corollary} \label{col:gap-mpu}
For any constants $c \geqs 2$ and $\alpha \in (0, c - 1)$, for $G\sim\cG_c(n,n^{-\alpha})$ and $p = n^{\alpha/(c - 1)}$, $\tilde\Omega(\log n)$ rounds of Sherali-Adams applied to a M$p$U relaxation with a lifted objective function have an integrality gap of $n^{\frac{\alpha(c - 1 - \alpha)}{(c - 1)(c - \alpha)}-o(1)}$ w.h.p. In particular, for $\alpha= c - \sqrt{c}$, we get an integrality gap of $n^{1 - 2/(1 + \sqrt{c}) - o(1)}$.
\end{corollary}

We remark that our integrality gaps for both D$k$SH and M$p$U match the log-density threshold. However, unlike their graph counterparts D$k$S and S$p$ES, no approximation algorithm whose ratios match the ones predicted by the log-density framework is known for D$k$SH and M$p$U for $c$-uniform hypergraphs where $c \geqs 4$. For $c = 3$, such algorithms are known only for the semi-random instances~\cite{CMinprep} but not for the general (worst case) instances.


The extension to super-constant edge size also has implications for the recently studied Small Set Bipartite Vertex Expansion (SSBVE) problem. In this problem, we are given a bipartite graph $G=(L,R,E)$ and a parameter $p$, and are asked to find $U\subseteq L$ of size $p$ with a minimum number of total neighbors in $R$. We also match the approximation guarantee of~\cite{CDM17} for SSBVE:
\begin{corollary} \label{col:gap-ssbve}
For any constant $\alpha\in(0,1)$, there is an infinite family of bipartite graphs, such that for any bipartite graph $G=(L,R,E)$ in this family, and $p=|L|^\alpha$, $\omega(1)$ rounds of Sherali-Adams applied to an SSBVE relaxation with a lifted objective function have an integrality gap of $|L|^{\alpha(1-\alpha) - o(1)}$.
\end{corollary}
Note that~\cite{CDM17} also gave a Sherali-Adams integrality gap for the above problem, though only for $\alpha=1/2$ (as in the integrality gap of~\cite{BCVGZ12}), and without lifting the objective function.

\subsection{Related Work}
As mentioned earlier, Sherali-Adams gaps for Densest $k$-Subgraph were studied in~\cite{BCVGZ12}. In their paper, they focused exclusively on the case of $\cG(n,n^{-\alpha})$ and $k=n^{\alpha}$ for $\alpha=1/2$. They defined an LP solution which, for every small vertex set  $S\subseteq V$, assigns to the variable $y_S$ some value which depends only on the size of the smallest Steiner tree for $S$. As they explain, such a solution is intuitive for the specific statistical properties of $\cG(n,n^{-1/2})$. However this intuition, and their analysis, break down for any other value of $\alpha$.

We also remark that, while lower bounds for the stronger Sum-of-Square (SoS) relaxations for D$k$S have been studied~\cite{BCVGZ12,CMMV17}, the integrality gaps given do not match the gaps predicted by log-density framework; in fact, the graphs used in these results are not random graphs but rather graphs resulting from a reduction from random instances of a constraint satisfaction problem (CSP). Hence, these SoS lower bounds are unrelated to the Planted Dense Subgraph Conjecture, and it is unlikely that the approach can yield similar gaps as predicted by the conjecture.

We turn instead to more recent work on integrality gaps for lift-and-project relaxations, namely the pseudo-calibration technique. This technique was introduced by Barak et al.~\cite{BHKKMP16}, to give tight Sum-of-Squares integrality gaps for the closely related Planted Clique problem. The pseudo-calibration heuristic involves choosing a pseudo-distribution whose low-degree Fourier coefficients match those of a random graph with a planted solution. This technique is a very powerful heuristic, however, it does not automatically guarantee that the pseudo-distribution suggested by the technique will satisfy the required constraints. It is not clear at this point whether such a solution for Densest $k$-Subgraph and related problems is feasible as a Sherali-Adams solution, much less Sum-of-Squares. Our solution, however, is inspired by pseudo-calibration. We are able to modify the solution suggested by pseudo-calibration in such a way that in fact it does satisfy at least the Sherali-Adams constraints. This connection is explained in Appendix~\ref{sec:pseudo-cal}.

\subsection{Our Technique}
In our analysis, we must show that three kinds of lifted constraints hold: monotonicity ($y_{T} \leqs y_{S}$ for any $S\subseteq T\subseteq V$), size constraints (bounding the number of vertices selected by the solution), and density constraints (bounding the number of edges among selected vertices).

As explained below, for every small subset $S \subseteq V$, our solution assigns to $y_S$ the maximum value attained for some function of $G|_{S'}$ over all possible vertex sets $S'\supseteq S$ up to a certain size bound. While at a first glance, both the monotonicity constraint and the density constraint seem to follow almost immediately from the definition, it turns out that this ``immediate'' proof of the density constraint relies on two different size bounds for $S'$. To reconcile the two conditions, we show that in fact in $\cG(n,n^{-\alpha})$, w.h.p.\ any optimal set $S'$ must already be bounded in size, thus making the difference in the requirements irrelevant.

The main technical aspect of our proof involves the size constraints. As we will show, a careful analysis of the properties of optimal sets does in fact guarantee the size constraints in expectation (over the random choices in $\cG(n,n^{-\alpha})$). However, this is not enough, since we also need to show strong concentration (especially since this must hold for all possible small subsets $S\subseteq V$). Unfortunately, the general kind of concentration result we need is not true for random graphs.\footnote{As an example of this issue, consider the number of copies of $H$ in $\cG(n,n^{-3/5})$, where $H$ is a clique of size 5 connected to a path of length 5. The expected number of copies of $H$ in $G$ is $n^{|V(H)|-3/5|E(H)|}=n^{10-3/5(15)}=n$. However, w.h.p.\ \emph{no} copies of $H$ will appear, since $H$ also contains $K_5$, which only has an expected $n^{5-3/5(10)}=1/n$ copies.} We again strongly rely on the properties of optimal sets, showing that they \emph{precisely} give the condition which guarantees concentration. This is shown by bounding the left-hand-side of the LP constraint by a polynomial in the incidence vector of $\cG(n,n^{-\alpha})$, and applying a concentration bound of Kim and Vu~\cite{KV00} for low-degree polynomials.

\subsection{Organization}
In the next section, we will describe our solution for the Sherali-Adams relaxation and show that it satisfies the three aforementioned properties; this is the main contribution of our paper. Then, in Section~\ref{sec:int-gap}, we state the relaxations for D$k$SH, M$p$U and SSBVE, and explain the parameters setting that give the claimed integrality gaps. Finally, we discuss several open questions in Section~\ref{sec:open}.


\section{Our Sherali-Adams Solution and Its Properties}
\label{sec:main-gen}

The goal of this section is to prove our main theorem, which is stated below.

\begin{theorem} \label{thm:main-general}
For any constants $0 < \beta, \varepsilon < 1$, any $c = c(n) = O(\log \log n)$ and any $\varepsilon < \alpha < c(n) - \varepsilon$, let $G = (V, E)$ be a hypergraph sampled from the Erd\H{o}s-R\'enyi random $c$-uniform hypergraph distribution $\cG_c(n,n^{-\alpha})$, then, w.h.p., there exists $\{y_S\}_{|S| \leqs r}$ for some $r = \tilde\Omega(\log n)$ that satisfies the following:
\begin{align}
\forall S, T \subseteq V \text{ such that } |S| + |T| \leqs r: \nonumber \\
\sum_{i \in V} \sum_{J \subseteq T} (-1)^{|J|} y_{S \cup J \cup \{i\}} &\leqs n^{\beta + o(1)} \sum_{J \subseteq T} (-1)^{|J|}y_{S \cup J} \label{eq:const-size-gen-1} \\
\sum_{e \in E} \sum_{J \subseteq T} (-1)^{|J|} y_{S \cup J \cup e} &\geqs n^{\beta(c - \alpha) - o(1)} \sum_{J \subseteq T} (-1)^{|J|} y_{S \cup J} \label{eq:const-edge-gen-1} \\
0 &\leqs \sum_{J \subseteq T} (-1)^{|J|} y_{S \cup J} \leqs 1 \label{eq:incl-excl-gen} \\
y_{\emptyset} &= 1.
\end{align}
\end{theorem}

To illustrate the above constraints, it is perhaps best to think about the Densest $k$-Subhypergraph problem as a concrete example. In this case, $y_i$ should be thought of as the indicator variable of whether the vertex $i$ is selected as part of the solution (i.e. the subset of $k$ vertices); more generally, $y_S$ should be thought of as an indicator variable whether $S$ is contained in the solution. 

Constraint~\eqref{eq:const-size-gen-1}, which we refer to as the \emph{size constraint} is the resulting lift of the constraint  (i.e., ``at most $k$ vertices are selected'') where $k = n^{\beta + o(1)}$. Now, recall that our objective is to maximize the number of induced hyperedges, i.e., to maximize $\sum_{e \in E} y_e$. Constraint~\eqref{eq:const-edge-gen-1} is the lift of the constraint $\sum_{e \in E} y_e \geqs k^{(c - \alpha) - o(1)}$. That is of the requirement that the solution must induce at least $k^{(c - \alpha) - o(1)}$ hyperedges; we call this constraint the \emph{density constraint}. In other words, the $y_S$'s given in Theorem~\ref{thm:main-general} behave as if they are a valid solution of the planted case of Conjecture~\ref{conj:planted-subgraph} where a random graph from $\cG_c(k, k^{-\alpha-o(1)})$ is planted into the instance. We remark also that in the typical LP relaxation, we also need to require that $y_{e} \geqs y_i$ for all $i \in e$, but this is already implied by Constraint~\eqref{eq:incl-excl-gen}, which results from lifting the constraint $0 \leqs y_S \leqs 1$. Section~\ref{sec:int-gap} describes in more details how a solution satisfying these constraints also provide integrality gaps for other problems including M$p$U and SSBVE.

We will now proceed to prove Theorem~\ref{thm:main-general}, starting by describing our solution.

\paragraph{The Solution.}
For all $S \subseteq V$ such that $|S| \leqs r(n)$, we define $y_S$ as follows:
\begin{align}
y_S = \frac{1}{L^{|S|}}\max_{\substack{S' \supseteq S \\ |S'| \leqs \tau(n)}} n^{(1 - \beta)(\talpha|E(S')| - |S'|)}
\label{eq:sol-def-gen}
\end{align}
where $L$ is some sub-polynomial dampening factor and $\talpha = \alpha - o(1)$ is a number slightly smaller than $\alpha$; we will define them more precisely later.

For brevity, for every $S \subseteq V$, let $\Phi(S)$ be $\talpha |E(S)| - |S|$, and let $\Phio(S) = \max_{S' \supseteq S, |S'| \leqs \tau(n)} \Phi(S')$. Moreover, we say that a subset $S'$ is a \emph{candidate set} for $S$ if $S' \supseteq S$ and $|S'| \leqs \tau(n)$. Note that our solution can be rewritten simply as $y_S = \frac{1}{L^{|S|}} \cdot n^{(1 - \beta)\Phio(S)}$.

For every $S \subseteq V$ of cardinality at most $r(n)$, we use $\cS'(S)$ to denote the collection of all candidate sets $S'$ of $S$ that attain the optimum in our solution $y_S$ defined in (\ref{eq:sol-def-gen}). In other words, a set $S' \supseteq S$ of size at most $\tau(n)$ belongs to $\cS'(S)$ if and only if $\Phi(S') = \Phio(S)$.

\paragraph{Parameter Selection.}
The exact parameters that we will use are as follows:
\begin{itemize}
\item $\tau = \tau(n) = \frac{\log n}{c(\log \log n)^2} = \tilde\Theta(\log n)$.
\item $r = r(n) = \frac{\tau(n)}{(\log \log n)^2} = \tilde\Theta(\log n)$.
\item $L = 2\tau(n) = \tilde\Theta(\log n)$.
\item $\talpha = \alpha\left(\frac{1 - 10r(n)/\tau(n)}{1 + 3c\tau(n)/\log n}\right) = \left(1 - O(1/(\log \log n)^2)\right)\alpha = \alpha - o(1)$.
\end{itemize}

\subsection{Dampening Factor and Simplified Constraints}

Our use of the dampening factor $L$ is exactly the same as that in~\cite{BCVGZ12}; namely, it will allow us to simplify all constraints to the case where $T = \emptyset$. To see this, first observe that our solution has the following strong monotonicity property:

\begin{lemma}
For every $S \subseteq V$ and every $i \in V$ such that $|S \cup \{i\}| \leqs r(n)$, we have $y_S \geqs L \cdot y_{S \cup \{i\}}$
\end{lemma}

\begin{proof}
Let $S'\in\cS'(S\cup\{i\})$. Since $S'$ is also a candidate set for $S$, we have $$y_{S}\geqs L^{-|S|}n^{(1-\beta)\Phi(S')}=L\cdot y_{S\cup\{i\}}.$$
\end{proof}

We can then arrive at the following lemma, which states that $\sum_{J \subseteq T} (-1)^{|J|} y_{S \cup J}$ is within a factor of two of $y_S$, as stated below. Since the proof is exactly identical to the proof\footnote{Bhaskara~\etal's proof~\cite{BCVGZ12} requires $L \geqs r(n)$; this also holds for our choice of parameter.} of Lemma 3.2 of~\cite{BCVGZ12}, we do not repeat it here.

\begin{lemma}[\cite{BCVGZ12}]
For any $S, T \subseteq V$ such that $|S \cup T| \leqs r$ and $S \cap T = \emptyset$, we have $$y_S \geqs \sum_{J \subseteq T} (-1)^{|J|} y_{S \cup J} \geqs y_S/2.$$
\end{lemma}

Hence, the above lemma readily implies~\eqref{eq:incl-excl-gen}. Moreover, as observed in~\cite{BCVGZ12}, we can ``replace'' each side in the constraints (\ref{eq:const-size-gen-1}) and (\ref{eq:const-edge-gen-1}) by the case where $T = \emptyset$ and lose a factor of at most two each time. (Note that when $S \cap T \ne \emptyset$, both sides of the constraints are already zero.) In other words, since these constant factors can be absorbed in the $o(1)$ term, it suffices to show that the two constraints hold when $T=\emptyset$. That is, we only need to show the following:
\begin{align}
\forall S \subseteq V \text{ such that } |S|  \leqs r: \nonumber \\
\sum_{i \in V} y_{S \cup \{i\}} &\leqs n^{\beta + o(1)} y_{S} \label{eq:size-const-gen} \\
\sum_{e \in E} y_{S \cup e} &\geqs n^{\beta(c - \alpha) - o(1)} y_{S} \label{eq:deg-const-gen} \\
y_{\emptyset} &= 1
\end{align}

\subsection{Subhypergraph Exceeding the Log-Density Threshold and the Size Cutoff}

Our analysis will rely heavily on the fact that w.h.p.\ random hypergraphs do not contain subhypergraphs with density strictly above the log-density threshold:

\begin{lemma}\label{lem:log-density-witness-gen} 
With high probability, $|E(S)| \leqs |S|\cdot \frac{(1+3c\tau/\log n)}{\alpha}$ for every $S \subseteq V$ of size at most $\tau$.
\end{lemma}
\begin{proof} Note that this bound holds trivially for any $S$ of size at most $c - 1$ since $|E(S)| = 0$. For succinctness, let $\tA=(1+3c\tau/\log n)/\alpha$. For every $k\in\{c,\ldots,\tau\}$, the number of labeled $k$-vertex hypergraphs $H$ with $t(k) = \lceil k\tA \rceil$ hyperedges is at most $\binom{k^c}{t(k)} \leqs k^{c \cdot t(k)}$. Since the expected number of copies of such an $H$ in $G$ is $n^{|V(H)|-\alpha|E(H)|} = n^{k - \alpha t(k)} \leqs 2^{-3c\tau k}$, we have
\begin{align*}
  \Pr[\exists H\text{ as above}:H\text{ is a subhypergraph of }G]&<\sum_{H\text{ as above}}\Ex[\#\text{copies of }H\text{ in G}]\\
&\leqs \sum_{k=c}^{\tau}\sum_{\substack{H\text{ as above}\\ |V(H)|=k}} 2^{-3c\tau k}\\
&\leqs \sum_{k=c}^{\tau} k^{c \cdot t(k)} 2^{-3c\tau k}\\
&= \sum_{k=c}^{\tau} 2^{c(t(k)\log k -3\tau k)}\\
&\leqs \sum_{k=c}^{\tau} 2^{c((1 + k\tA)\log k -3\tau k)} \\
&\leqs \sum_{k=c}^{\tau} 2^{ck(\tA\log k -2\tau)} \\
&\leqs \sum_{k=c}^{\tau} 2^{ck(\tA\log \tau - 2 \tau)}.
\end{align*}
Now, observe that, for sufficiently large $n$, we have $\tA\log \tau \leqs \tau$. As a result, we have
\begin{align*}
  \Pr[\exists H\text{ as above}:H\text{ is a subgraph of }G]
  \leqs \sum_{k=c}^{\tau} 2^{-ck\tau} 
  \leqs 2^{1 -c^2 \tau}
  \leqs 2^{-c^2 \tau / 2}
  = o(1),
\end{align*}
which concludes our proof.
\end{proof}

Note that since we select $\talpha < \frac{\alpha}{1 + 3c\tau(n)\log n}$, we can immediately conclude that, with probability $1 - o(1)$, $\Phi(S) \leqs 0$ for all $S \subseteq V$ of size at most $\tau(n)$, which immediately imply that $y_{\emptyset} = 1$:
\begin{corollary}
With high probability, $y_{\emptyset} = 1$.
\end{corollary}

More importantly, Lemma~\ref{lem:log-density-witness-gen} implies that the value of the cutoff $\tau(n)$ in fact does not matter! Specifically, as shown below, all sets $S'$ that attain the maximum value of $S$ in fact has size at most $\tau'(n) = \tau(n) / 10$. Note that in the proof of the following lemma we need a separation between $\talpha$ and $\alpha$, which is the main reason behind our choice of $\talpha$. 

\begin{lemma}\label{lem:opt-set-size-gen}
With high probability, $\cS'(S) \subseteq \binom{V}{\leqs \tau'(n)}$ for all $S \in \binom{V}{\leqs r(n)}$.
\end{lemma}

\begin{proof}
Let us assume that the high probability event from Lemma~\ref{lem:log-density-witness-gen} occurs. Suppose for the sake of contradiction that there is some $S'\in\cS'(S)$ of size $|S'|\in[\tau(n)/10,\tau(n)]$. Then in particular, $S'$ is at least as good a set for $S$ as $S$ itself. That is, we have
\begin{align*}
0 &\leqs \Phi(S') - \Phi(S) \\
&= \talpha(|E(S')| - |E(S)|) - (|S'| - |S|) \\
&\leqs \talpha|E(S')| - (1 - |S| / |S'|) |S'| \\
&\leqs \talpha|E(S')|-(1-10r(n)/\tau(n))|S'|.
\end{align*}
But this implies that
$$|E(S')|\geqs |S'|\cdot(1-10r(n)/\tau(n))/\talpha \geqs |S'|\cdot(1+3c\tau/\log n)/\alpha,$$ where the second inequality follows from our choice of $\talpha$. This contradicts our assumption that the high probability event in Lemma~\ref{lem:log-density-witness-gen} occurs.
\end{proof}

\subsection{Proof of The Density Constraint}

We will next prove constraint~\eqref{eq:deg-const-gen}. Assume that the high probability event in Lemma~\ref{lem:opt-set-size-gen} holds. Let $S' \in \cS'(S)$. By our assumption, $|S'| \leqs \tau'(n)$, which is at most $\tau(n) - c$ when $n$ is sufficiently large. Hence, for every hyperedge $e \in E$, the set $S' \cup e$ is a candidate set for $S \cup e$, meaning that if $e\not\subseteq S'$, then
\begin{align*}
y_{S \cup e} \geqs L^{-|S|-c} n^{(1 - \beta)\Phi(S' \cup e)} \geqs 
L^{-|S|-c} n^{(1 - \beta)(\Phi(S')+\alpha'-c)} = L^{-c} n^{(1 - \beta)(\talpha - c)} y_S.
\end{align*}

Moreover, observe that w.h.p.\ the total number of hyperedges in $G$ (and not contained in $S'$) is at least $(1 - o(1)) \cdot n^{c - \alpha}$. As a result, summing the above inequality over all $e \in E$ yields the following:
\begin{align*}
\sum_{e \in E} y_{S \cup e} &\geqs (1 - o(1)) L^{-c} n^{(c - \alpha) + (1 - \beta)(\talpha - c)} y_S \\
&= (1 - o(1)) L^{-c} n^{\beta(c - \alpha) - (1 - \beta)(\alpha - \talpha)} y_S \\
&= n^{\beta(c - \alpha) - o(1)}y_S.
\end{align*}

\paragraph{Degree constraints} Algorithms in the log-density framework often rely on a bound on the minimum degree in (part of) an optimal solution. We note that when $\alpha< c-1-\epsilon$ (that is, when $G$ is dense enough not to have isolated vertices), our solution also satisfies the corresponding lifted constraints for every vertex. As before, due to the dampening factors, these constraints are equivalent to the following:
\begin{align}
  \sum_{e \in E:e\ni i} y_{S \cup e} &\geqs n^{\beta(c - 1-\alpha) - o(1)} y_{S} &\forall S\subseteq V\text{ such that }|S|\leqs r, \forall i\in S
\end{align}
The proof follows much the same argument. As above, for $S' \in \cS'(S)$, node $i\in S(\subseteq S')$, and hyperedge $e\not\subseteq S'$ that contains $i$, it follows by a similar calculation that $y_{S\cup e}\geqs L^{-c+1}n^{\alpha'-c+1}$. Combining this bound with the observation that w.h.p.\ the number of hyperedges containing $i$ is at least $(1-o(1))n^{c-1-\alpha}$, we obtain the above constraint.

\subsection{Proof of The Size Constraint}
\label{sec:size-gen}

We now turn our attention to the only remaining constraint: the size constraint (Constraint (\ref{eq:size-const-gen})). The proof of the size constraint is somewhat more complicated than the previous constraints and, before we can get to the main argument, we will need to prove a couple more structural properties of the optimal candidate sets.

\subsubsection{Structural Properties of Optimal Candidate Sets}

First, observe that $\Phi(\cdot)$ is supermodular; this will be useful in the sequel.

\begin{observation} \label{obs:supmod}
For every $S_1, S_2 \subseteq V$, $\Phi(S_1 \cup S_2) + \Phi(S_1 \cap S_2) \geqs \Phi(S_1) + \Phi(S_2)$.
\end{observation}

\begin{proof}
Observe that $|E(S_1 \cup S_2)| + |E(S_1 \cap S_2)| \geqs |E(S_1)| + |E(S_2)|$ and
$|S_1 \cup S_2| + |S_1 \cap S_2| = |S_1| + |S_2|$. Subtracting the two yields the above bound.
\end{proof}

From this point on, we will assume that the high probability event in Lemma~\ref{lem:opt-set-size-gen} occurs, i.e., that all optimal candidate sets (for all $S \subseteq V$ of size at most $r(n)$) are of size at most $\tau'(n)$.

We call a collection $\cS \subseteq \cP(V)$ \emph{union closed}, if, for every $S_1, S_2 \in \cS$, $S_1 \cup S_2 \in \cS$. The supermodularity of $\Phi$ and the fact that cutoff does not matter implies the union-closedness of $\cS'(S)$:

\begin{lemma}\label{lem:union-closed-gen}
For all $S \subseteq V$ of size at most $r(n)$, $\cS'(S)$ is union closed.
\end{lemma}

\begin{proof}
Consider any set $S_1, S_2 \in \cS'(S)$. Since $S \subseteq S_1, S_2$ and $|S_1|, |S_2| \leqs \tau'(n)$, $S_1 \cap S_2$ must also be a candidate set for $S$, which implies that $\Phi(S_1 \cap S_2) \leqs \Phio(S) = \Phi(S_1) = \Phi(S_2)$. As a result, from Observation~\ref{obs:supmod}, we have $\Phi(S_1 \cup S_2) \geqs \Phi(S_1) + \Phi(S_2) - \Phi(S_1 \cap S_2) \geqs \Phio(S)$. Furthermore, $|S_1 \cup S_2| \leqs |S_1| + |S_2| \leqs 2\tau'(n) \leqs \tau(n)$, meaning that $S_1 \cup S_2$ is a candidate set for $S$ as well. Hence, $S_1 \cup S_2$ must also be in $\cS'(S)$.
\end{proof}

Union closed collection of sets always contains a maximal set, which is simply the union of all the sets. When $\cS'(S)$ is union closed, we denote its maximal set by $S'_{\max}(S)$. The following lemma relates the maximal optimal candidate set of $S \cup \{i\}$ to that of $S$:

\begin{lemma}
For all $S \subseteq V$ of size at most $r(n)$ and for all $i \in V$, $S'_{\max}(S) \subseteq S'_{\max}(S \cup \{i\})$.
\end{lemma}

\begin{proof}
The proof is essentially the same as that of Lemma~\ref{lem:union-closed-gen}. Consider any $S'' \in \cS'(S\cup\{i\})$. Since $S'_{\max}(S)$ is optimal for $S$ and $S'_{\max}(S) \cap S''$ is also a candidate set for $S$, we have $\Phi(S'_{\max}(S)) \geqs \Phi(S'_{\max}(S) \cap S'')$. Similar to the calculation in the previous lemma, the supermodularity of $\Phi$ implies that $\Phi(S'_{\max}(S) \cup S'') \geqs \Phi(S'')$. Moreover, since $|S'_{\max}(S) \cup S''| \leqs 2\tau'(n) \leqs \tau(n)$, $S'_{\max}(S) \cup S''$ is a valid candidate set for $S$, meaning that $(S'_{\max}(S) \cup S'') \in \cS'(S \cup \{i\})$ as desired.
\end{proof}

\subsubsection{Setting up the Argument}

We are now ready to bound $\sum_{i \in V} y_{S \cup \{i\}}$. For brevity, we write $S_{\max}$ to denote $S'_{\max}(S)$ and $S^{i}_{\max}$ to denote $S'_{\max}(S \cup \{i\})$ for every $i \in V$. We can bound $\sum_{i \in V} y_{S \cup \{i\}}$ as follows.
\begin{align}
\sum_{i \in V} y_{S \cup \{i\}}
&= \sum_{i \in S_{\max}} y_{S \cup \{i\}} + \sum_{i \notin S_{\max}} y_{S \cup \{i\}} \\
&= |S_{\max}| \cdot \frac{y_S}{L} + \sum_{i \notin S_{\max}} \frac{1}{L^{|S| \cup \{i\}}} \cdot \frac{n^{(1-\beta)\talpha|E(S^i_{\max})|}}{n^{(1-\beta)|S^i_{\max}|}} \nonumber \\
&\leqs y_S \left(\frac{\tau(n)}{L} + \sum_{i \notin S_{\max}} \frac{n^{(1-\beta)\talpha|E(S^i_{\max}) \setminus E(S_{\max})|}}{n^{(1-\beta)|S^i_{\max} \setminus S_{\max}|}}\right) \\
&\leqs y_S \left(1 + \sum_{i \notin S_{\max}} \frac{n^{(1-\beta)\talpha|E(S^i_{\max}) \setminus E(S_{\max})|}}{n^{(1-\beta)|S^i_{\max} \setminus S_{\max}|}}\right) \label{eq:size-bound-1-gen} 
\end{align}


Let $\cI_{\iso}$ be the set of $i \in V$ such that $i$ is an isolated vertex in the subhypergraph $E(S^i_{\max})$. It is not hard to show the following:

\begin{proposition}
\label{prop:iso}
\begin{enumerate}
\item For every $i \in \cI_{\iso} \setminus S_{\max}$, we have $S_{\max} \cup \{i\} \in \cS'(S \cup \{i\})$. \label{prop:iso-1}
\item For every $i \notin \cI_{\iso}$, no vertex outside of $S$ is isolated in the subhypergraph $E(S^i_{\max})$. \label{prop:non-iso}
\end{enumerate}
\end{proposition}

\begin{proof}
\begin{enumerate}
\item Suppose for the sake of contradiction that $i \in \cI_{\iso}\setminus S_{\max}$ but $(S_{\max} \cup \{i\}) \notin \cS'(S \cup \{i\})$. This implies
\begin{align*}
\Phio(S) = \Phi(S_{\max}) \leqs \Phi(S_{\max} \cup \{i\}) - 1 < \Phi(S^i_{\max}) - 1 = \Phi(S^i_{\max} \setminus \{i\}) \leqs \Phio(S),
\end{align*}
where the second inequality comes from $S_{\max} \cup \{i\} \notin \cS'(S \cup \{i\})$ and the next equality comes from our assumption that $i$ is an isolated vertex in $E(S^i_{\max})$ and the last inequality follows from the fact that $S^i_{\max} \setminus \{i\}$ is a candidate set for $S$. Thus, we have arrived at a contradiction.
\item Suppose for the sake of contradiction that $i \notin \cI_{\iso}$ but there exists $j \in S^i_{\max} \setminus S$ such that $j$ is isolated in $E(S^i_{\max})$. From $i \notin \cI_{\iso}$, we have $j \ne i$; this means that $S^i_{\max} \setminus \{j\}$ is a candidate set for $S \cup \{i\}$ but this implies that
\begin{align*}
\Phio(S \cup \{i\}) = \Phi(S^i_{\max}) = \Phi(S^i_{\max} \setminus \{j\}) - 1 \leqs \Phio(S \cup \{i\}) - 1,
\end{align*}
which is a contradiction.
\end{enumerate}
\end{proof}

From property~\ref{prop:iso-1} of Proposition~\ref{prop:iso}, we can rewrite~\eqref{eq:size-bound-1-gen} further as follows:
\begin{align}
\sum_{i \in V} y_{S \cup \{i\}} &\leqs y_S \left(1 + |\cI_{\iso}| \cdot \frac{1}{n^{1 - \beta}} + \sum_{i \notin S_{\max} \cup \cI_{\iso}} \frac{n^{(1-\beta)\talpha|E(S^i_{\max}) \setminus E(S_{\max})|}}{n^{(1-\beta)|S^i_{\max} \setminus S_{\max}|}}\right) \\
&\leqs y_S \left(1 + n^{\beta} + \sum_{i \notin S_{\max} \cup \cI_{\iso}} \frac{n^{(1-\beta)\talpha|E(S^i_{\max}) \setminus E(S_{\max})|}}{n^{(1-\beta)|S^i_{\max} \setminus S_{\max}|}}\right)\label{eq:size-bound-2}
\end{align}

To bound the remaining sum in the above right hand side expression, we will need some additional notations and observations.

First, observe that, from the optimality of $S^i_{\max}$ for $S \cup \{i\}$, $\Phi(S^i_{\max})$ must be at least $\Phi(S_{\max} \cup \{i\}) \geqs \Phi(S_{\max}) - 1$, which implies the following:

\begin{observation}
For every $i \notin S_{\max}$,  $\talpha |E(S^i_{\max}) \setminus E(S_{\max})| - |S^i_{\max} \setminus S_{\max}| \geqs -1.$
\end{observation}

Next, we define an $\talpha$-strictly balanced ($\talpha$-s.b.) sets of hyperedges with respect to $S_{\max}$. Intuitively speaking, a set $T$ of hyperedges is $\talpha$-s.b. with respect to $S_{\max}$ if we cannot select a subset $T' \subseteq S_{\max}$ such that we can add $T'$ into $S_{\max}$ and create a new set with higher $\Phi$ than $\Phi(S_{\max})$. This is formalized below.

\begin{definition}
For each $T \subseteq \binom{V}{c}$ such that $T \cap \binom{S_{\max}}{c} = \emptyset$, $T$ is said to be \emph{$\talpha$-strictly balanced} ($\talpha$-s.b.) with respect to $S_{\max} \subseteq V$ if for any $T' \subseteq T$, we have $|v(T') \setminus S_{\max}| > \talpha |T'|$.
\end{definition}

The optimality of $S_{\max}$ for $S$ implies the following:
\begin{observation}
For every $T \subseteq E_G$ such that $T \cap E(S_{\max}) = \emptyset$ and $|v(T) \setminus S_{\max}| \leqs \tau(n) - \tau'(n)$, $T$ is $\talpha$-strictly balanced with respect to $S_{\max}$.
\end{observation}

For convenience, we say that for $S^* \subseteq V$, a hypergraph $H$ with vertices $S^*$ satisfies condition (*) if the following conditions hold:
\begin{itemize}
\item $S^* \supseteq S_{\max}$,
\item $E(H) \setminus E(S_{\max})$ is $\talpha$-strictly balanced with respect to $S_{\max}$,
\item $\talpha |E_H(S^*) \setminus E(S_{\max})| - |S^* \setminus S_{\max}| \geqs - 1$
\item Every $v \in S^* \setminus S_{\max}$ is not isolated in $E_H(S^*)$.
\end{itemize}

With the above two observations and the second property of Proposition~\ref{prop:iso}, we can further bound the remaining sum in~(\ref{eq:size-bound-2}) as follows.

\begin{align*}
&\sum_{i \notin S_{\max} \cup \cI_{\iso}} \frac{n^{(1-\beta)\talpha|E(S^i_{\max}) \setminus E(S_{\max})|}}{n^{(1-\beta)|S^i_{\max} \setminus S_{\max}|}} \nonumber \\
&\leqs \sum_{\substack{S^* \in \binom{V}{\leqs \tau'(n)}\\ G|_{S^*} \text{ satisfies } (*)}} \sum_{\substack{i \notin S_{\max} \cup \cI_{\iso}\\S^i_{\max} = S^*}} \frac{n^{(1-\beta)\talpha|E(S^*) \setminus E(S_{\max})|}}{n^{(1-\beta)|S^* \setminus S_{\max}|}} \nonumber \\
&\leqs \tau'(n) \left(\sum_{\substack{S^* \in \binom{V}{\leqs \tau'(n)}\\ G|_{S^*} \text{ satisfies } (*)}} \frac{n^{(1-\beta)\talpha|E(S^*) \setminus E(S_{\max})|}}{n^{(1-\beta)|S^* \setminus S_{\max}|}}\right) \nonumber \\
\end{align*}

Since $\tau'(n)=n^{o(1)}$, we will focus on bounding the final sum

\begin{equation}
\sum_{\substack{S^* \in \binom{V}{\leqs \tau'(n)}\\ G|_{S^*} \text{ satisfies } (*)}} \frac{n^{(1-\beta)\talpha|E(S^*) \setminus E(S_{\max})|}}{n^{(1-\beta)|S^* \setminus S_{\max}|}}
\label{eq:size-bound-3}
\end{equation}

As we shall see, in expectation, this expression is bounded from above by $n^{\beta + o(1))}$, as required.
 However, this is not enough for us: we want this sum to be bounded by $n^{\beta + o(1)}$ for all choices of $S$'s simultaneously. To do so, we will first prove a concentration for a fixed $S$ and then use a union bound over all $n^{O(r(n))}$ choices of $S$'s.

We begin with an intuitive overview of the proof, before giving a more formal proof in Section~\ref{sec:concentration-formal}.


\paragraph{Bounding the Expectation}
To bound~\eqref{eq:size-bound-3} in expectation, we consider every possible hypergraph of size $\tau'(n)$ which could appear as a subhypergraph containing $S_{\max}$. Consider a hypergraph ``template'' $H_{X,T}=(S_{\max}\cup X,T\cup E(S_{\max}))$, where $X$ are (at most) $\tau'(n)-|S_{\max}|$ dummy vertices, and $T\subseteq \binom{S_{\max}\cup X}{c}\setminus\binom{S_{\max}}{c}$. Then~\eqref{eq:size-bound-3} can be bounded from above by
\begin{equation}\sum_{\substack{X,T\text{ as above}\\ H_{X,T} \text{ satisfies } (*)}} \frac{n^{(1-\beta)\alpha'|T|}\cdot|\{S^*\supseteq S_{\max}:H_{X,T}\text{ is isomorphic to a subgraph of }G|_{S^*}\}|}{n^{(1-\beta)|X|}}
\label{eq:size-bound-4}
\end{equation}

To bound the expectation (under the random choices of $\cG_c(n,n^{-\alpha})$), it suffices to bound the expected number of copies of each $H_{X,T}$ in the above expression. This expectation is easily seen to be $n^{|X|-\alpha|T|}\leq n^{|X|-\alpha'|T|}$, allowing us to bound the expectation of~\eqref{eq:size-bound-4} by
\begin{align*}
  \sum_{\substack{X,T\text{ as above}\\ H_{X,T} \text{ satisfies } (*)}} \frac{n^{(1-\beta)\alpha'|T|}}{n^{(1-\beta)|X|}}\cdot\frac{n^{|X|}}{n^{\alpha'|T|}} &=  \sum_{\substack{X,T\text{ as above}\\ H_{X,T} \text{ satisfies } (*)}} \frac{n^{\beta|X|}}{n^{\beta\alpha'|T|}}\\
&=\sum_{\substack{X,T\text{ as above}\\ H_{X,T} \text{ satisfies } (*)}} n^{\beta(|X|-\alpha'|T|)}\\
&\leqs\sum_{\substack{X,T\text{ as above}\\ H_{X,T} \text{ satisfies } (*)}} n^{\beta}. &\text{by the last condition in (*)}
\end{align*}

Since the number of possible template hypergraphs $H_{X,T}$ of size $\tau'(n)$ is again subpolynomial, in expectation, we have our bound of $n^{\beta+o(1)}$.

\paragraph{Proving Concentration around the Expectation.}
As mentioned earlier, we will formally prove that concentration holds by expressing~\eqref{eq:size-bound-3} as a low-degree polynomial, and applying a concentration bound of Kim and Vu~\cite{KV00}. However, this application is not straightforward. Not only is a large expectation not enough for our analysis, it is not true for general sets that we would even necessarily have any meaningful concentration. We must use the fact that $S_{\max}$ is an optimal set. Specifically, that we only consider $\alpha$-strictly balanced subgraphs containing $S_{\max}$.

To see why this is important, let us consider a template hypergraph $H_{X,T}$ for $S_{\max}$. Recall that the expected number of copies of $H_{X,T}$ in $G$ is $n^{|X|-\alpha|T|}$. While it may be that $|X|>\alpha|T|$, giving us a polynomially large expectation, it is still possible (from a purely graph-theoretic perspective) that for some set $X'$ s.t.\ $S_{\max}\subseteq X'\subset X$, the induced subtemplate $H_{X',T'}$ (where $T'=T\cap E(H|_{X'\cup S_{\max}})$) does not satisfy this, but rather has $|X'|<\alpha|T'|$.  In this case, rather than concentration, we would have a case where with high probability there is no copy of $H_{X,T}$ (since in particular w.h.p.\ there is no copy of $H_{X',T'}$). This does not seem bad for us, since there are even far fewer copies than expected, and we would like to get an upper bound. However, the flip side of this situation is that with $n^{-\Omega(1)}$ probability, there are actually polynomially many times \emph{more} copies of $H_{X,T}$ than the expectation, and a polynomially small probability is not good enough, since our bound must hold simultaneously for all $n^{\Omega(r(n))}$ small sets $S$.

Fortunately, the $\alpha$-balanced condition precisely states that for any subtemplate $H_{X',T'}$ we in fact have $|X'|\geqs\alpha|T'|$, so this situation will not arise. Thus concentration is not ruled out, and in fact can be formally proven using the Kim-Vu bound, as we will now show.

\subsubsection{Concentration of Low Degree Multilinear Polynomials.}\label{sec:concentration-formal}

Recall that we would like to bound~\eqref{eq:size-bound-3}. For each $e \in \binom{V}{c}$, let $X_e$ denote the indicator variable $\ind[e \in E_G]$ for $G=\cG_c(n,n^{-\alpha})$. Then~\eqref{eq:size-bound-3} can easily be bounded by a multilinear polynomial in these random variables $X_e$'s as follows:
\begin{align}
\sum_{\substack{S^* \in \binom{V}{\leqs \tau'(n)}\\ G|_{S^*} \text{ satisfies } (*)}} \frac{n^{(1-\beta)\talpha|E(S^*) \setminus E(S_{\max})|}}{n^{(1-\beta)|S^* \setminus S_{\max}|}}
&\leqs \sum_{S^* \in \binom{V}{\leqs \tau'(n)}} \frac{1}{n^{(1-\beta)|S^* \setminus S_{\max}|}} \cdot 
\left(\sum_{\substack{T \subseteq \binom{S^*}{c} \setminus \binom{S_{\max}}{c} \\ 
(S^*,T)\text{ satisfies (*)}
}}
\ind[T \subseteq E_G] \cdot n^{(1-\beta)\talpha|T|}\right)\nonumber\\
&= \sum_{\substack{T \subseteq \binom{S^*}{c} \setminus \binom{S_{\max}}{c} \\ 
(S^*,T)\text{ satisfies (*)}
}} \frac{n^{(1 - \beta)\talpha|T|}}{n^{(1 - \beta)|v(T) \setminus S_{\max}|}} \prod_{e \in T} X_e\label{eq:F-def}
\end{align}
Denote the polynomial in~\eqref{eq:F-def} by $F((X_e)_{e \in \binom{V}{c}})$. Note that $F$ has degree $t \leqs \tau'(n)/\talpha$ since $|v(T) \cup S_{\max}| \leqs \tau'(n)$ and $T$ being strictly balanced implies that $|T| \leqs |v(T) \setminus S_{\max}| / \talpha \leqs \tau'(n) / \talpha$.

To show that this quantity is not too large, we will resort to the following concentration bound of low degree multilinear polynomial of Kim and Vu~\cite{KV00}, which applies for any low degree polynomials whose \emph{expectation} of every partial derivative is small. We note that this condition is weaker than those of some other similar inequalities, such as Azuma's~\cite{AZ}, which requires the \emph{maximum} (of partial derivative or some other measures of effect) to be small.


\begin{theorem}[Kim-Vu Concentration Bound~\cite{KV00}] \label{thm:KV}
Let $f(x_1, \dots, x_N)$ be any degree-$t$ multilinear polynomial, i.e.,
\begin{align*}
f(x_1, \dots, x_N) = \sum_{T \in \cT} w_T \prod_{i \in T} x_i
\end{align*}
where $\cT$ is a collection of subsets of $[N]$ of size at most $t$ and $w_T \geqs 0$ is a non-negative weight of each $T \in \cT$.

For each $A \subseteq [N]$, let the truncated polynomial $f_A(x_1, \dots, x_N)$ be defined as
\begin{align*}
f_A(x_1, \dots, x_N) = \sum_{\substack{T \in \cT \\ A \subseteq T}} w_T \prod_{i \in T \setminus A} x_i.
\end{align*}
(In other words, for every monomial containing $A$, we substitute $x_i = 1$ for all $i \in A$. And other monomials are deleted entirely from the sum. Note that $f_A$ does not depend on $x_i$ for any $i \in A$.)

For any independent Bernoulli random variables $X_1, \dots, X_N$, let $E_A$ denote $\Ex[f_A(X_1, \dots, X_N)]$ for every $A \subseteq [N]$. Moreover, let $E$ denote $\max_{A \subseteq [N]} E_A$. Then, for any $\lambda > 1$, we have
\begin{align*}
\Pr[|f(X_1, \dots, X_N) - E_{\emptyset}| > E (8\lambda)^t \sqrt{t!}] < 2e^2 e^{-\lambda} n^{t - 1}.
\end{align*}
\end{theorem}

To apply the Kim-Vu bound, we only need to show that $E$ is also $n^{\beta + o(1)}$ as the term that is multiplied with $E$ is negligible. The bound on $E$ is stated below; we remark that both the $\talpha$-strictly balancedness and the constraint $\talpha|T| - |v(T) \setminus S_{\max}| \geqs - 1$ are needed. The latter is needed even when bounding the expectation of $F$ (i.e. $E_{\emptyset}$). On the other hand, roughly speaking, the former is used to ensure that the set $T'$ that we take the partial derivative on is not too dense; otherwise, it could increase the expectation significantly. This intuition is formalized in the following proof.

\begin{lemma}
For $F((X_e)_{e \in \binom{V}{c}})$ as defined above and for any (possibility empty) subset $T' \subseteq \binom{V}{c}$, we have $E_{T'} \leqs \left(\tau'(n)^{c\tau'(n)/\talpha} \tau'(n)^2 / \talpha \right) k$. In particular, this also implies that $E \leqs \left(\tau'(n)^{c\tau'(n)/\talpha} \tau'(n)^2 / \talpha \right) k$.
\end{lemma}

\begin{proof}
Observe that $F_{T'}((X_e)_{e \in \binom{V}{c}})$ is exactly 
\begin{align*}
\sum_{\substack{T' \subseteq T \subseteq \binom{V}{c} \setminus \binom{S_{\max}}{c} \\ 
(S^*,T)\text{ satisfies (*)}
}}
\frac{n^{(1 - \beta)\talpha|T|}}{n^{(1 - \beta)|v(T) \setminus S_{\max}|}} \prod_{e \in T \setminus T'} X_e.
\end{align*}
Hence, if $|v(T') \cup S_{\max}| \geqs \tau'(n)$, then this is simply zero. Furthermore, since the sum is only over $T$ that is $\talpha$-strictly balanced, the expression is also zero unless $|v(T') \setminus S_{\max}| > \talpha|T'|$. As a result, from now on, we may assume that $|v(T') \setminus S_{\max}| > \talpha|T'|$. For brevity, let $y = |T'|$ and $z = |v(T') \setminus S_{\max}|$.

Now, observe that $\Ex[\prod_{e \in T \setminus T'} X_e]$ is simply $n^{-\alpha|T \setminus T'|}$. In other words, we have
\begin{align*}
E_{T'} &= \sum_{\substack{T' \subseteq T \subseteq \binom{V}{c} \setminus \binom{S_{\max}}{c} \\ 
(S^*,T)\text{ satisfies (*)}
}}
\frac{n^{(1 - \beta)\talpha|T|}}{n^{(1 - \beta)|v(T) \setminus S_{\max}| + \alpha|T\setminus T'|}} \\
&\leqs \sum_{\substack{a, b \in \mathbb{N}_0 \\ b + z + |S_{\max}| \leqs \tau'(n) \\ a \leqs \tau'(n) / \talpha \\ \talpha(a + y) - (b + z) \geqs -1}} \sum_{\substack{T' \subseteq T \subseteq \binom{V}{c} \setminus \binom{S_{\max}}{c} \\ |v(T) \setminus (v(T') \cup S_{\max})| = b \\ |T| = a + y}} \frac{n^{(1 - \beta)\talpha(a + y)}}{n^{(1 - \beta)(b + z) + \alpha a}} \\
&\leqs \sum_{\substack{a, b \in \mathbb{N}_0 \\ b + z + |S_{\max}| \leqs \tau'(n) \\ a \leqs \tau'(n) / \talpha\\ \talpha(a + y) - (b + z) \geqs -1}} n^b \cdot \binom{\tau'(n)^c}{a} \cdot \frac{n^{(1 - \beta)\talpha(a + y)}}{n^{(1 - \beta)(b + z) + \alpha a}} \\
&\leqs \tau'(n)^{c\tau'(n)/\talpha} \cdot \sum_{\substack{a, b \in \mathbb{N}_0 \\ b + z + |S_{\max}| \leqs \tau'(n) \\ a \leqs \tau'(n) / \talpha \\ \talpha(a + y) - (b + z) \geqs -1}} n^b \cdot \frac{n^{(1 - \beta)\talpha(a + y)}}{n^{(1 - \beta)(b + z) + \talpha a}} \\
&= \tau'(n)^{c\tau'(n)/\talpha} \cdot \sum_{\substack{a, b \in \mathbb{N}_0 \\ b + z + |S_{\max}| \leqs \tau'(n) \\ a \leqs \tau'(n) / \talpha \\ \talpha(a + y) - (b + z) \geqs -1}} n^{\beta((b + z) - \talpha(a + y))} \cdot n^{\talpha y - z} \\
&< \tau'(n)^{c\tau'(n)/\talpha} \cdot \sum_{\substack{a, b \in \mathbb{N}_0 \\ b + z + |S_{\max}| \leqs \tau'(n) \\ a \leqs \tau'(n) / \talpha \\ \talpha(a + y) - (b + z) \geqs -1}} n^\beta \cdot n^0 \\
&\leqs \left(\tau'(n)^{c\tau'(n)/\talpha} \tau'(n)^2 / \talpha \right) k.
\end{align*}
\end{proof}

\subsubsection{Putting Things Together}

Note that for sufficiently large $n$, we have $\tau'(n) \geqs 1/\talpha$. Applying Theorem~\ref{thm:KV} with the above bound and with $\lambda = (\log n)^{10}$ and $t = \tau'(n) / \talpha$, we arrive at the following bound:
\begin{align*}
\Pr\left[F((X_e)_{e \in \binom{V}{2}}) > \left(2(8\log n)^{10 \tau'(n)/\talpha}(\tau'(n))^{3 + (c + 1)\tau'(n)/\talpha}\right) \cdot k\right] \leqs O(e^{-\log^{10} n + (\tau'(n)/\talpha)\log n})
\end{align*}

Together with (\ref{eq:size-bound-1-gen}) and (\ref{eq:size-bound-2}), the above equation implies
\begin{align*}
\Pr\left[\frac{\sum_{i \in V} y_{S \cup \{i\}}}{y_S} > \left(2 + 2(8\log n)^{10 \tau'(n)/\talpha}(\tau'(n))^{3 + (c + 1)\tau'(n)/\talpha}\right) \cdot k \right] \leqs O(e^{-\log^{10} n + (\tau'(n)/\talpha)\log n}).
\end{align*}

Hence, by taking union bound over all $S$ of size at most $r(n) \leqs \tau'(n)$ we have
\begin{align*}
\Pr\left[\exists S, \frac{\sum_{i \in V} y_{S \cup \{i\}}}{y_S} > \left(2 + 2(8\log n)^{10 \tau'(n)/\talpha}(\tau'(n))^{3 + (c + 1)\tau'(n)/\talpha}\right) \cdot k \right] \leqs O(e^{-\log^{10} n + (2\tau'(n)/\talpha)\log n}).
\end{align*}

Finally, observe that in our parameter selection $c \tau'(n) = o(\log n / \log \log n)$; hence, the term $\left(2 + 2(8\log n)^{10 \tau'(n)/\talpha}(\tau'(n))^{3 + (c + 1)\tau'(n)/\talpha}\right)$ is $n^{o(1)}$ as desired, and $O(e^{-\log^{10} n + (2\tau'(n)/\talpha)\log n}) = e^{-\log^{10} n + o(\log^2 n)} = o(1)$. In other words, have shown that the size constraint holds for all $S \subseteq V$ of size at most $r(n)$ with high probability, which concludes our proof.


\section{The Integrality Gaps}
\label{sec:int-gap}

With Theorem~\ref{thm:main-general} ready, it is simple to arrive at the integrality gaps; since the solution from Theorem~\ref{thm:main-general} is already mimicking a graph with planted $\cG_c(n^{\beta+o(1)}, n^{-\beta\alpha-o(1)})$, we can just compare the optimal LP value there with the actual solution in $G \sim \cG_c(n, n^{-\alpha})$.

\subsection{Densest $k$-Subhypergraph}
\label{sec:int-gap-dksh}

First, recall that the standard LP relaxation for Densest $k$-Subhypergraph can be stated as the LP on the left hand side below; here $y_i$ (respectively, $y_e$) is intended to denote whether the vertex $i$ (respectively, hyperedge $e$) belongs to the selected subhypergraph.

\begin{center}
\begin{tabular}{c c c | c c c}
maximize & $\sum_{e\in E} y_e$ & &  maximize & $p$ & \\
subject to & $\sum_{i \in V} y_i\leqs k$ & & subject to & $\sum_{i \in V} y_i\leqs k$ & \\
& & & & $\sum_{e \in E} y_e \geqs p$ \\
& $y_e \leqs y_i$ & $\forall e \in E, i \in e$ & & $y_e \leqs y_i$ & $\forall e \in E, i \in e$ \\
 & $0 \leqs y_i \leqs 1$ & $\forall i \in V$ & & $0 \leqs y_i \leqs 1$ & $\forall i \in V$
\end{tabular}
\end{center}

The right hand side LP is an equivalent formulation of the left hand side LP where we also added the density constraint that the total number of hyperedges must be at least $p$, and then put $p$ as an objective. Note that here we think of $p$ as a fixed number rather than a variable in this LP. (The program can be solved by a binary search on $p$.) Lifting this LP using the $r$-level Sherali-Adams hierarchy yields the following relaxation for Densest $k$-Subhypergraph:
\begin{align}
\text{maximize } & p \nonumber \\
\text{subject to } \nonumber \\
\forall S, T \subseteq V \text{ such that } |S| + |T| \leqs r: \nonumber \\
\sum_{i \in V} \sum_{J \subseteq T} (-1)^{|J|} y_{S \cup J \cup \{i\}} &\leqs k \sum_{J \subseteq T} (-1)^{|J|}y_{S \cup J} \nonumber \\
\sum_{e \in E} \sum_{J \subseteq T} (-1)^{|J|} y_{S \cup J \cup e} &\geqs p \sum_{J \subseteq T} (-1)^{|J|} y_{S \cup J} \label{eq:lifted-density-bound} \\
0 &\leqs \sum_{J \subseteq T} (-1)^{|J|} y_{S \cup J} \leqs 1 \label{eq:lifted01} \\
y_{\emptyset} &= 1. \nonumber
\end{align}

\begin{remark} Lifting the objective function here is crucial. Without constraint~\eqref{eq:lifted-density-bound}, nothing prevents us from setting $y_S=0$ for all $|S|\geqs 3$ and obtaining a much larger integrality gap for many more rounds. LP-based algorithms for D$k$S also use a locally lifted objective function. Indeed, in the previous work of Bhaskara~\etal~\cite{BCVGZ12} a similar, but slightly different, constraint is lifted. Specifically, they lift the ``degree constraint'' that every vertex in the selected subgraph has degree at least $\Omega(p/k)$. While our density constraint does not imply the degree constraint, it is easy to see the degree constraint is also satisfied in our solution (via essentially the same proof as that of the density constraint).
\end{remark}

\begin{remark}
The constraint $y_e \leqs y_i$ does not appear in the lifted LP since it is already implied by~\eqref{eq:lifted01}.
\end{remark}

Our solution from Theorem~\ref{thm:main-general} already satisfies the constraints in the above LP. Indeed, by selecting appropriate parameters, we immediately arrive at the claim integrality gaps for D$k$SH.

\begin{proof}[Proof of Corollary~\ref{col:gap-dksh}]
We select $\beta = \alpha/(c - 1)$. From Theorem~\ref{thm:main-general}, there exists a solution for $\tilde\Omega(\log n)$-level Sherali-Adams LP above with $p = n^{\beta(c - \alpha) - o(1)}$ and $k = n^{\beta + o(1)}$. On the other hand, it is easy to see that, w.h.p., any $k$-vertex subhypergraph contains at most $\tilde{O}(\max\{k, k^cn^{-\alpha}\}) = \tilde{O}(n^{\beta + o(1)})$ hyperedges. Hence, this is an integrality gap of $n^{\beta(c - 1 - \alpha) - o(1)} = n^{\alpha(1 - \alpha/(c - 1)) - o(1)}$ as desired.
\end{proof}

\subsection{Minimum $p$-Union}

The standard relaxation for Minimum $p$-Union is the left hand side LP below. Again, we will work with the equivalent LP on the right, which yields stronger lifts in the Sherali-Adams hierarchy.

\begin{center}
\begin{tabular}{c c c | c c c}
minimize & $\sum_{i \in V} y_i$ & &  minimize & $k$ & \\
subject to & & & subject to & $\sum_{i \in V} y_i\leqs k$ & \\
& $\sum_{e \in E} y_e \geqs p$ & & & $\sum_{e \in E} y_e \geqs p$ \\
& $y_e \leqs y_i$ & $\forall e \in E, i \in e$ & & $y_e \leqs y_i$ & $\forall e \in E, i \in e$ \\
 & $0 \leqs y_i \leqs 1$ & $\forall i \in V$ & & $0 \leqs y_i \leqs 1$ & $\forall i \in V$
\end{tabular}
\end{center}

The $r$-round Sherali-Adams lift of the right hand side LP is exactly the same as that of D$k$SH:
\begin{align*}
\text{minimize } & k \\
\text{subject to } \\
\forall S, T \subseteq V \text{ such that } |S| + |T| \leqs r: \\
 \sum_{i \in V} \sum_{J \subseteq T} (-1)^{|J|} y_{S \cup J \cup \{i\}} &\leqs k \sum_{J \subseteq T} (-1)^{|J|}y_{S \cup J} \\
\sum_{e \in E} \sum_{J \subseteq T} (-1)^{|J|} y_{S \cup J \cup e} &\geqs p \sum_{J \subseteq T} (-1)^{|J|} y_{S \cup J} \\
0 &\leqs \sum_{J \subseteq T} (-1)^{|J|} y_{S \cup J} \leqs 1 \\
y_{\emptyset} &= 1.
\end{align*}

Again, we arrive at integrality gaps by plugging appropriate parameters into Theorem~\ref{thm:main-general}.

\begin{proof}[Proof of Corollary~\ref{col:gap-mpu}]
We select $\beta = \frac{\alpha}{(c - 1)(c - \alpha)}$. From Theorem~\ref{thm:main-general}, there exists a solution for $\tilde\Omega(\log n)$-level Sherali-Adams LP above with $p = n^{\beta(c - \alpha) - o(1)}$ and $k = n^{\beta + o(1)}$. On the other hand, w.h.p., any subhypergraph containing $p$ hyperedges must consist of at least $\tilde{\Omega}(\min\{p, (pn^{\alpha})^{1/c}\}) \geqs n^{\beta(c - \alpha) - o(1)}$ vertices. Hence, this yields an integrality gap of $n^{\beta(c - 1 - \alpha) - o(1)} = n^{\frac{\alpha(c - 1 - \alpha)}{(c - 1)(c - \alpha)} - o(1)}$.
\end{proof}

\subsection{Small Set Bipartite Vertex Expansion}

The standard LP relaxation for SSBVE is the left hand side LP below. Again, we work with its equivalence on the right, which yields stronger lifts in the Sherali-Adams hierarchy.

\begin{center}
\begin{tabular}{c c c | c c c}
minimize & $\sum_{i \in R} y_i$ & &  minimize & $k$ & \\
subject to & & & subject to & $\sum_{i \in R} y_i\leqs k$ & \\
& $\sum_{j \in L} y_j \geqs p$ & & & $\sum_{j \in L} y_j \geqs p$ \\
& $y_i \geqs y_j$ & $\forall (j, i) \in E$ & & $y_i \geqs y_j$ & $\forall (j, i) \in E$ \\
 & $1 \geqs y_i \geqs 0$ & $\forall i \in V$ & & $1 \geqs y_i \geqs 0$ & $\forall i \in V$
\end{tabular}
\end{center}

The $r$-round Sherali-Adams lift of the right hand side LP above can be stated as follows.
\begin{align*}
\text{minimize } & k \\
\text{subject to } \\
\forall S, T \subseteq L \cup R \text{ such that } |S| + |T| \leqs r: \\
 \sum_{i \in R} \sum_{J \subseteq T} (-1)^{|J|} y_{S \cup J \cup \{i\}} &\leqs k \sum_{J \subseteq T} (-1)^{|J|}y_{S \cup J}\\
\sum_{j \in L} \sum_{J \subseteq T} (-1)^{|J|} y_{S \cup J \cup \{j\}} &\geqs p \sum_{J \subseteq T} (-1)^{|J|} y_{S \cup J} \\
 \sum_{J \subseteq T} (-1)^{|J|} y_{S \cup J \cup \{j\}} &\leqs \sum_{J \subseteq T} (-1)^{|J|} y_{S \cup J \cup \{i\}}  &\forall (j, i) \in E\\
0 &\leqs \sum_{J \subseteq T} (-1)^{|J|} y_{S \cup J} \leqs 1 \\
y_{\emptyset} &= 1.
\end{align*}

\begin{proof}[Proof of Corollary~\ref{col:gap-ssbve}]
Observe that if we view a $c$-uniform hypergraph as the incidence bipartite graph where the hyperedge are vertices on the left and the original vertices are vertices on the right. Then, a solution to $\tilde{\Omega}(\log n)$-round of Minimum $p$-Union immediately gives a solution to SSBVE with the same number of rounds and the same values of $p$ and $k$.

Hence, we can use the same analysis as in the proof of Corollary~\ref{col:gap-mpu}, but this time we pick $\beta$ to be slightly higher than $\gamma$, i.e., $\beta = \gamma + \frac{1}{\log \log n}$, $c = \log \log n$ and $\alpha = \frac{c(c - 1)\beta}{1 + (c - 1)\beta}$. The solution satisfies the constraint for $p = n^{\beta(c - \alpha) - o(1)} \geqs n^{\gamma(c - \alpha) + o(1)} \geqs |L|^\gamma$ where the last inequality comes from the fact that $L$ concentrates around $n^{c - \alpha}$. Finally, with a similar calculation as in the previous section, the integrality gap is $n^{\beta(c - 1 -\alpha) - o(1)} = |L|^{\beta(1 - \frac{1}{c - \alpha}) - o(1)}$. Observe that $\frac{1}{c - \alpha} = \frac{1 + (c - 1)\beta}{c} \leqs \beta + o(1)$. Thus, the integrality gap is at least $|L|^{\beta(1 - \beta) - o(1)} = |L|^{\gamma(1 - \gamma) - o(1)}$ as claimed.
\end{proof}



\section{Discussions and Open Questions}
\label{sec:open}
Having shown a general technique for proving Sherali-Adams gaps at the log-density threshold, it would now be interesting to see if our approach can be applied to more rigidly structured problems that have been considered in the log-density framework, such as Label Cover.

Furthermore, an exciting challenge that could further bolster the conjectured hardness would be to give matching lower-bounds in the Sum-of-Squares SDP hierarchy. However, note that for some parameter regimes, such an integrality gap does not hold even for a simple SDP! As shown in~\cite{BCCFV10}, for $1>\alpha>1/2$, a simple SDP relaxation for D$k$S applied to $G=\cG(n,n^{-\alpha})$ can already witness the fact that $G$ does not contain a $k=n^{\alpha}$-subgraph with more than $n^{(1+\alpha)/2}\ll k^{2-\alpha}$ edges. On the other hand, in the worst case, the best currently-known approximation for this value of $k$ is still $k^{1-\alpha}$ (matching the log-density threshold). Thus, at least for certain parameter regimes, it would be very interesting to prove either of the following two possibilities, both of which are consistent with the current state of our knowledge:
\begin{enumerate}
  \item There are algorithms for certain parameter regimes which give (worst-case) approximation guarantees strictly better than the log-density gap.
  \item There is some other family of instances (other than Erd\H{o}s R\'enyi (hyper)graphs) for which there is a hardness of approximation, or at least Sum-of-Squares integrality gap, matching the log-density threshold.
\end{enumerate}

\subsubsection*{Acknowledgments}

PM thanks Tselil Schramm and Euiwoong Lee for insightful discussions on related questions.

\bibliography{main}
\bibliographystyle{alpha}

\appendix

\section{Pseudo-Calibration for Densest $k$-Subhypergraph}
\label{sec:pseudo-cal}

Let $n, k \in \N$ and $p, q \in [0, 1]$ be parameters to be chosen later.
\begin{itemize}
\item Let $\Grd = \cG_c(n, p)$ denote the distribution of Erdos-Renyi random $c$-uniform hypergraph on $n$ vertices where each $c$-size set of vertices is a hyperedge w.p. $p$.
\item Let $\Gpl$ denote the ``planted'' distribution where each of the $n$ vertices is marked as in the planted solution independently with probability $k/n$. For each $c$-size set of vertices, if all of its elements are marked as in the planted solution, then they are joined by a hyperedge with probability $q$. Otherwise, they are joined with probability $p$.
\end{itemize}

We represent a graph $G$ as $\{\frac{1 - p}{\sqrt{p(1 - p)}}, \frac{-p}{\sqrt{p(1 - p)}}\}^{\binom{[n]}{c}}$ where $G_e = \frac{1 - p}{\sqrt{p(1 - p)}}$ if the hyperedge $e$ exists and in the graph and $G_e = \frac{-p}{\sqrt{p(1 - p)}}$ otherwise. As usual, for each $T \subseteq \binom{[n]}{c}$, let $\chi_T(G) = \prod_{e \in T} G_e$. Observe that
\begin{align*}
\Ex_{G \sim \Grd}[\chi_T(G) \chi_{T'}(G)] = 
\begin{cases}
1 & \text{if } T = T', \\
0 & \text{otherwise.}
\end{cases}
\end{align*}

Finally, we use $x_u \in \{0, 1\}$ to denote the indicator variable of whether $u \in [n]$ is marked as planted. As usual, we define $x_S = \prod_{u \in S} x_u$

The Pseudo-Calibration heuristic suggests the following pseudo-distribution as a candidate solution.
\begin{align*}
\tE[x_S](G) = \sum_{T \subseteq \binom{[n]}{c} \atop |v(T) \cup S| \leqs \tau} \widehat{{\tE[x_S]}}(T) \chi_T(G)
\end{align*} 
where
\begin{align*}
\widehat{{\tE[x_S]}}(T) = \Ex_{(x, G) \sim \Gpl}[x_S\chi_T(G)] 
&= \left(\frac{k}{n}\right)^{|v(T) \cup S|} \Ex_{(x, G) \sim \Gpl}[\chi_T(G) \mid x_{v(T) \cup S} = 1] \\
&= \left(\frac{k}{n}\right)^{|v(T) \cup S|} \left(\frac{q - p}{\sqrt{p(1 - p)}}\right)^{|T|} \\
&\approx \left(\frac{k}{n}\right)^{|v(T) \cup S|} \left(\frac{q}{\sqrt{p}}\right)^{|T|}.
\end{align*}

In other words, the suggested solution is
\begin{align*}
\tE[x_S](G) = \sum_{T \subseteq \binom{[n]}{2} \atop |v(T) \cup S| \leqs \tau} \left(\frac{k}{n}\right)^{|v(T) \cup S|} \left(\frac{q}{\sqrt{p}}\right)^{|T|} \chi_T(G).
\end{align*} 

We will further approximate this by replacing $\chi_T(G)$ with $(1/\sqrt{p})^{|T|}\ind[T \subseteq E_G]$. This gives the solution
\begin{align*}
\tE[x_S](G) = \sum_{T \subseteq E_G \atop |v(T) \cup S| \leqs \tau} \left(\frac{k}{n}\right)^{|v(T) \cup S|} \left(\frac{q}{p}\right)^{|T|}.
\end{align*} 

Our solution (specified in~\eqref{eq:sol-def-gen}) is the same as the above expression except that we change the summation to maximization, and we add the dampening factor and some $o(1)$ slack in the exponent; note that the parameters there are $p = n^{\alpha}, k = n^\beta, q = k^{\alpha} = n^{\alpha\beta}$.

\end{document}